\RequirePackage{lineno}
\documentclass[jmp,superscriptaddress,showpacs]{revtex4}

\usepackage[T1]{fontenc}
\usepackage{color,graphicx}
\usepackage{times}
\usepackage[colorlinks=true]{hyperref}
\usepackage{amsmath,amsthm,amssymb}
\usepackage{setspace}

\newcommand\ket[1]{\ensuremath{|#1\rangle}}
\newcommand\bra[1]{\ensuremath{\langle#1|}}

\newcommand\tr{\mathop{\rm Tr}\nolimits}
\newcommand\Tr{\mathop{\rm Tr}\nolimits}

\newcommand\rank{\mathop{\rm rank}\nolimits}

\newcommand\supp{\mathop{\rm supp}\nolimits}

\newtheorem{theorem}{Theorem}
\newtheorem{corollary}{Corollary}
\newtheorem{remark}{Remark}

\newtheorem{problem}{Problem}
\newtheorem{example}{Example}

\begin{document}

\ifdefined\lineno
    \linenumbers
\else
\fi


\title{Rank Reduction for the Local Consistency Problem}

\author{Jianxin Chen}%
\affiliation{Department of Mathematics \& Statistics, University of
  Guelph, Guelph, Ontario, Canada}%
\affiliation{Institute for Quantum Computing, University of Waterloo,
  Waterloo, Ontario, Canada}%
\author{Zhengfeng Ji}%
\affiliation{Institute for Quantum Computing, University of Waterloo,
  Waterloo, Ontario, Canada}%
\affiliation{State Key Laboratory of Computer Science, Institute of Software, Chinese Academy of Sciences, Beijing, China}%
\author{Alexander Klyachko}
\affiliation{Department of Mathematics, Bilkent University, Bilkent, Ankara, Turkey}
\author{David W. Kribs}%
\affiliation{Department of Mathematics \& Statistics, University of
  Guelph, Guelph, Ontario, Canada}%
 \affiliation{Institute for Quantum Computing, University of Waterloo,
  Waterloo, Ontario, Canada}%
\author{Bei Zeng}%
\affiliation{Department of Mathematics \& Statistics, University of
  Guelph, Guelph, Ontario, Canada}%
\affiliation{Institute for Quantum Computing, University of Waterloo,
  Waterloo, Ontario, Canada}%

\begin{abstract}
We address the problem of how simple a solution can be for a given quantum local consistency instance. More specifically, we investigate how small the rank of the global density operator can be if the local constraints are known to be compatible. We prove that any compatible local density operators can be satisfied by a low rank global density operator. Then we study both fermionic and bosonic versions of the $N$-representability problem as applications. After applying the channel-state duality, we prove that  any compatible local channels can be obtained through a global quantum channel with small Kraus rank.
\end{abstract}

\date{\today}


\maketitle

\section{Introduction} \label{sec:pre}

Understanding the various correlations and relationships amongst the different parts of a many body quantum system is one of the most difficult challenges in quantum  theory.  It is well known that the reduced density operators defined by partial traces characterize subsystems. Consider a system of three parties $A$, $B$ and $C$: If all three two-particle density operators $\rho^{AB}$, $\rho^{AC}$ and $\rho^{BC}$ are consistent with some global density operator $\rho^{ABC}$, they must satisfy $\tr_A(\rho^{AB})=\tr_C(\rho^{BC})$, $\tr_B(\rho^{AB})=\tr_C(\rho^{AC})$ and $\tr_A(\rho^{AC})=\tr_B(\rho^{BC})$. This is a necessary but not sufficient condition for the existence of $\rho^{ABC}$. As the particle number $N$ becomes very large, the correlations between local subsystems become much more complicated. In general, local consistency is the problem of deciding whether a given collection of subsystem descriptions is consistent with some state of the global system, or the problem of finding necessary and sufficient condition for consistency of subsystem descriptions. It is also called the quantum marginal problem in literature\cite{K04,K05, AK08}.
The community observed the relation between the spectrum of bipartite quantum states and certain representations of the symmetric groups very recently. The consistency conditions for some special classes of quantum states were then given in \cite{K04,K05,AK08, CM06}. For general states with overlapping margins, the situation remains unclear.

If the particles under consideration are fermions instead of qubits, the local consistency problem has been known as the $N$-representability problem, which arose initially in the 1960's in connection with calculating the ground-state energies of general interacting electrons\cite{C63}.

It was only recently shown in \cite{L06, LCV07}  that both deciding the local consistency  and deciding the local consistency for fermions are QMA-complete, meaning both the consistency problem and the $N$-representability problem are computationally at least as hard as any other problem in the complexity class QMA. Here, QMA is the quantum analogue of the complexity class NP. Consequently, it is unlikely to have efficient algorithms for local consistency problems, even on a quantum computer. And very recently, Wei, et al. proved that the bosonic version of the $N$-representability problem is also QMA-complete\cite{WMN10}.

Though the local consistency problem is theoretically hard in the worst case, it is still worth exploring the potential solutions. There are various approaches scattered through the literature on this subject. Linden, et al. proved that almost every three-qubit pure state can be uniquely determined among all states by their two-party reduced states\cite{LW02}. A related  fermionic version was discussed in \cite{S65}, whose results indicate that almost any three fermion pure state is uniquely determined among all states by their two-particle reduced states, though it was not stated explicitly.  Linden et al.'s result was later generalized to $N$-particle systems\cite{JL05}.

In this paper, we will focus on another direction of the local consistency problem. We are interested in how simple the solution can be, or more specifically,  how small the rank of a  solution can be. The same question for bipartite quantum system without overlapping margins was discussed in \cite{K04}, but their approach seems technically difficult to generalize to the case with overlapping margins. In this work, we consider this problem regarding the rank of the solution in a very general setting, for multipartite quantum systems with overlapping margins. We provide a rank reduction based approach. We also show that some useful results from convex analysis can be applied to this problem directly, though it leads to a slightly weaker bound than ours\footnote{ See Remark~\ref{convex} for details.}. Then we will apply our results to ferminonic and bosonic systems. Finally local consistency problem for quantum channels will be addressed.

We now state our main result. For a given finite-dimensional Hilbert space $\mathcal{H}$, $B(\mathcal{H})$ will denote the space of bounded linear operators acting on $\mathcal{H}$. For the $n$-fold tensor product $\mathcal{H}^{{\otimes n}}$ and any integer $i\leq n$, $X_{i}$ and $Z_{i}$ are general Pauli $X$-gate and general Pauli $Z$-gate for $i$-th qudit respectively. Formally, the local consistency problem can be stated as follows.

\begin{problem}
Consider a multipartite quantum system $\mathcal{A}_1\mathcal{A}_2\cdots \mathcal{A}_n$ with Hilbert space $\mathcal{H}_{\mathcal{A}_{1}\mathcal{A}_{2}\cdots \mathcal{A}_{n}}=\mathcal{H}_{\mathcal{A}_{1}}\otimes \mathcal{H}_{\mathcal{A}_{2}}\otimes \cdots \otimes \mathcal{H}_{\mathcal{A}_{n}}$.  Given a set of reduced density operators $\rho_i(i=1,2,\cdots,k)$  where each $\rho_i$ acts on a subsystem $I_{i}$ of $\mathcal{A}_1\mathcal{A}_2\cdots\mathcal{A}_n$, or $I_{i}\subseteq \{\mathcal{A}_{1}, \mathcal{A}_{2},\cdots,\mathcal{A}_{n}\}$. The local consistency problem is to address the existence of a global density operator $\rho \in B(\mathcal{H}_{\mathcal{A}_{1}\mathcal{A}_{2}\cdots \mathcal{A}_{n}})$ satisfying $\tr_{I_{i}^{c}}\rho=\rho_i$ for any $1\leq i\leq k$. $I_{i}^{c}$ here is the complement of $I_{i}$.
\end{problem}

An \emph{instance} of the local consistency problem is a collection of pairs: a reduced density operator $\rho_{i}$ and corresponding subsystem $I_{i}$. For any given instance $\{(\rho_{i}, I_{i})\}_{i=1}^{k}$ of the local consistency problem, if there is a global density operator $\rho \in B(\mathcal{H}_{\mathcal{A}_{1}\mathcal{A}_{2}\cdots \mathcal{A}_{n}})$ satisfying $\tr_{I_i^{c}}\rho=\rho_i$ for any $1\leq i\leq k$, then we say $\{(\rho_{i}, I_{i})\}_{i=1}^{k}$ is a \textit{compatible} instance.

In this paper, we will show that  compatible local density operators are reduced states of some simple (low rank) global density operator. More specifically, we have following theorem.

\begin{theorem}\label{thm:main}
For any compatible instance of local consistency problem $\{(\rho_{i}, I_{i})\}_{i=1}^{k}$,  there is a solution with rank no more than
$\textstyle \sqrt{\sum_{i=1}^{k}{(\rank\rho_{i})^2}}$.
\end{theorem}

We are primarily interested in instances where the whole system is $\mathcal{H}^{\otimes n}$ and only no more than $c$-particle reduced states are known. In this case, the number of reduced states should be no more than ${n \choose c}$, the rank is bounded by polynomial in $n$ while the rank of a general density operator should be exponential in $n$.

The paper is structured as follows. After introducing the requisite background material  in section~\ref{sec:pre}, we give a proof of above theorem in section~\ref{sec:proof}. We then apply these ideas to the $N$-representability problem in section~\ref{sec:nrep}. In section~\ref{sec:channel} we introduce the notion of consistency of quantum channels.  Some examples are illustrated in section~\ref{sec:example}.

\section{Proof of Main Theorem}\label{sec:proof}
\begin{proof}[Proof of Theorem~\ref{thm:main}]
Since the local density operators of this instance are known to be compatible, we can start from some solution $\rho$, which is a density operator acting on $\mathcal{H}_{\mathcal{A}_{1}\mathcal{A}_{2}\cdots \mathcal{A}_{n}}$ and $\tr_{I_i^{c}}\rho=\rho_i$  for any $1\leq i\leq k$. Using a spectral decomposition of $\rho$ we may write
\begin{eqnarray}
\rho=\sum_{j=1}^r p_j\ket{\psi_j}\bra{\psi_j}
\end{eqnarray}
where $r=\rank(\rho)$ and $\ket{\psi_{1}},\ket{\psi_{2}},\cdots,\ket{\psi_{r}}$ are mutually orthogonal unit vectors.  We will show that when $r$ is large, then we can find another solution $\rho^{\prime}$ to the same instance and $\rank(\rho^{\prime})<\rank(\rho)$.

Consider the spectral decompositions again and write

\begin{eqnarray}
\rho_{i}=\sum_{j=1}^{r_{i}} p_{j}^{(i)} \ket{\psi_{j}^{(i)}}\bra{\psi_{j}^{(i)}} \textrm{ \  for\ all \ }  1\leq i\leq k
\end{eqnarray}
where $r_{i}={\rank} \rho_{i}$ and $\ket{\psi_{1}^{(i)}},\ket{\psi_{2}^{(i)}},\cdots,\ket{\psi_{r_i}^{(i)}}$ are mutually orthogonal unit vectors for any $i$.

Consider the following set:
\begin{align}
\mathcal{M} =\{X\in B(\supp \rho):  \Pi_{i}(\tr_{I_i^c}X) \Pi_{i}=O  \textrm{\ for\ any\ }  1\leq i\leq k\}.
\end{align}
Here, for each $i$, $\Pi_{i}$ is a projector on $\supp \rho_{i}$.

Now, $\mathcal{M}$ is a subspace of dimension at least $r^2-\sum_{i=1}^{k}{\rank^2}{\Pi_i}$.

If $r^2-\sum_{i=1}^{k}{\rank^2}{\Pi_i}\geq 1$, then
${\mathcal{M}}$ is not empty. Let's say there is a non-zero $X\in \mathcal{M}$  which implies $X^{\dagger}\in \mathcal{M}$ too. Thus both $H_1=X+X^{\dagger}$ and $H_2=i(X-X^{\dagger})$ are traceless Hermitian operators in $\mathcal{M}$. Note that $H_1$ and $H_2$ cannot be zero simultaneously when $X$ is not the zero operator. Without loss of generality, let us assume $H=H_{1}(\textrm{or\ } H_{2})$ is non-zero.

$H$ is chosen from $B(\supp \rho)$, or equivalently, there is some $\epsilon$ such that $\rho\pm \epsilon H \geq 0$ which follows $\rho_i\pm \epsilon \tr_{I_i^c}(H) \geq 0$ for any $i$. 
Thus Hermitian operator  $\tr_{I_i^c}(H)$ lies completely in $B(\supp \rho_{i})$  that implies $\Pi_{i}(\tr_{I_i^c}H) \Pi_{i}=\tr_{I_i^c}H$ and then $\tr H=\tr \tr_{I_i^c}H=0$.

Since the operator $H\ne 0, \Tr H=0$
contains both positive and negative eigenvalues, the same holds for
$\rho-\lambda H$ for $\lambda>>1$. Hence there exists an intermediate value
$0<\lambda<\infty$ for which the operator $\rho-\lambda H$ is nonnegative,
but not strictly positive, i.e. $\rho-\lambda H$ is a degenerate density
matrix we are looking for.


Now, for any solution $\rho$ to a given instance $\{\rho_{i},\mathcal{I}_{i}\}_{i=1}^{k}$, if
$r^2-\sum_{i=1}^{k}{\rank^2}{\Pi_i}\geq 1$, we can always find a non-zero traceless Hermitian operator $H\in \mathcal{M}$,and then another solution $\rho^{\prime}=\rho-\lambda H$  to the same instance with rank less than $\rank\rho$ . Thus , by repeating this procedure until above quadratic inequality doesn't hold anymore, we will finally end with a solution $\sigma$ with $\rank\sigma \leq \lfloor\sqrt{\sum_{i=1}^{k}{\mathop{(\rank\rho_{i})^2}}}\rfloor$.
\end{proof}

\begin{corollary}
For any instance of the local consistency problem with given local systems $\{I_{i}\}_{i=1}^{k}$, if the solution set is nonempty, then there is a solution with rank no more than $\lfloor\sqrt{\sum_{i=1}^{k}{{(\dim{I}_i)^2}}} \rfloor$.
\end{corollary}

\begin{remark}\label{convex}
Barvinok proved that if there is a positive semidefinite matrix $X$ satisfying
\begin{eqnarray}
\langle Q_{i}, X \rangle =q_{i}, \qquad\forall 1\leq i\leq k
\end{eqnarray}
where $Q_{1},\cdots,Q_{k}$ are symmetric matrices and $q_{1},\cdots, q_{k}$ are complex numbers, then there is a positive semidefinite matrix $X^{\ast}$ satisfying the same equation system and additionally $\rank X^{\ast}\leq \lfloor \frac{\sqrt{8k}-1}{2}\rfloor$\cite{B95}.
The main ingredients of his proof are the duality for linear programming in the quadratic form space. After applying Barvinok's theorem to the local consistency problem, we will have a similar rank reduction which will lead to a solution with rank no more than $\sqrt{2\sum_{i=1}^{k}{{\dim^2}{I}_i}}$. Thus this result is weaker than ours.
\end{remark}


\section{Application: $N$-Representability Problem} \label{sec:nrep}
In this section, we will study the $N$-representability problem, which is a fermionic analogue of the local consistency problem. The bosonic version of $N$-representability is also addressed later.

We first restate the $N$-representability problem as follows.

\begin{problem}
Given a system of $N$ fermions where each particle has $d$ energy levels, and a $k$-fermion state $\rho$ of size ${d \choose k}\times {d\choose k}$, determine
whether there exists an $N$-fermion state $\sigma$ such that $\tr_{k+1,\cdots,N}(\sigma)=\rho$.
\end{problem}

According to the Pauli exclusion principle, no two particles can occupy the same state, thus we can always assume $d\geq N$.

The space of $N$-fermion pure states is mathematically described as the $N$-th antisymmetric tensor product  of $\mathbb{C}_{d}$ with dimension ${d\choose N}$, denoted as $\wedge^{N}\mathbb{C}_{d}$. It is the span of all $N$-fold antisymmetric tensor products of vectors $x_{1}$, $x_{2}$, $\cdots$, $x_{N}$ in $\mathbb{C}_{d}$ which is defined as
\begin{eqnarray}
x_{1}\wedge x_{2}\wedge \cdots \wedge x_{N}=\frac{1}{\sqrt{N!}}\sum_{P}\varepsilon_{P}x_{P(1)}\otimes x_{P(2)}\otimes \cdots \otimes x_{P(N)}.
\end{eqnarray}
Here, $P$ goes through all permutations of $N$ indices and $\varepsilon_{P}$ is the signature of $P$.  So $\varepsilon_{P}$ is $1$, if the number of even-order cycles in $P$'s cycle type is even, and $-1$ otherwise. 

Similarly, the space of $N$-boson pure states with $d$ energy levels corresponds to the $N$-th symmetric tensor product of $\mathbb{C}_{d}$ with dimension $N+d-1 \choose N$, denoted as $\vee^{N}\mathbb{C}_{d}$.  

For more information about $N$-th symmetric/antisymmetric tensor product, please refer to \cite{B97}.

For the $N$-representability problem, there is a similar rank reduction as follows.

\begin{theorem}\label{thm:fermion}
Suppose we are given a system of $N$ fermions where each particle has $d$ energy levels and a $k$-fermion density operator $\rho$ of size ${d \choose k}\times {d\choose k}$.
Assume there exists an $N$-fermion state $\sigma$ such that $\tr_{k+1,\cdots,N}(\sigma)=\rho$. Then there also exists an $N$-fermion density operator
$\sigma^{\prime}$  with  $\tr_{k+1,\cdots,N}(\sigma^{\prime})=\rho$ and $\rank{\sigma^{\prime}}\leq  \rank \rho\leq {d\choose k}$.
\end{theorem}

The proof is similar to the proof provided in Section~\ref{sec:proof}, with minor modifications. Observe that the whole rank reduction in our approach is processed in $\supp(\rho)$. After introducing additional symmetry to the global system, the rank reduction also works by replacing $\otimes ^{N} \mathbb{C}_{d}$ with $\wedge^{N}\mathbb{C}_{d}$ or $\vee^{N}\mathbb{C}_{d}$.

Similarly, we will also get the following theorem for the bosonic version.
\begin{theorem}\label{thm:boson}
Suppose we have a system of $N$ bosons where each particle has $d$ energy levels and a $k$-boson density operator $\rho$ of size ${d+k-1 \choose k}\times {d+k-1\choose k}$.
Assume there exists an $N$-boson state $\sigma$ such that $\tr_{k+1,\cdots,N}(\sigma)=\rho$,
then there also exists an $N$-boson density operator
$\sigma^{\prime}$  with  $\tr_{k+1,\cdots,N}(\sigma^{\prime})=\rho$ and $\rank{\sigma^{\prime}}\leq  \rank \rho \leq {d+k-1\choose k}$ .
\end{theorem}



\section{Local Consistency Problem for Quantum Channels}\label{sec:channel}

In this section, we will investigate a new type of consistency -- the consistency of quantum channels. A quantum channel is a device which transmits classical bits or quantum states. Mathematically, it is a linear map which maps any quantum state on some Hilbert space $\mathcal{H}_{1}$ to another state on some Hilbert space $\mathcal{H}_{2}$. Furthermore, a quantum channel can be described by a completely-positive, trace-preserving map $\Phi$.

Generally, for a quantum channel from system $\mathcal{H}_{1}$ to $\mathcal{H}_{2}$, we shall think of $\mathcal{H}_1$ as part of a closed composite system $\mathcal{H}_1\otimes \mathcal{H}_1^{\prime}$ and $\mathcal{H}_2$ as part of another closed composite system $\mathcal{H}_2\otimes \mathcal{H}_2^{\prime}$ which has same dimension as $\mathcal{H}_1\otimes \mathcal{H}_1^{\prime}$. Therefore, the evolution from $\mathcal{H}_1\otimes \mathcal{H}_1^{\prime}$ to $\mathcal{H}_2\otimes \mathcal{H}_2^{\prime}$ can be described by some unitary operator $U$. The quantum channel $\Phi$ is then described as
\begin{eqnarray}
\Phi(\rho)=\tr_{\mathcal{H}_2^{\prime}}(U(\rho\otimes \ket{0}\bra{0}_{\mathcal{H}_1^{\prime}})U^{\dagger}).
\end{eqnarray}

By Stinespring's dilation theorem on completely positive maps,
$\Phi$ must take following form
\begin{eqnarray}
\Phi(X)=\sum_{i=1}^N K_i A K_i^{\dagger}
\end{eqnarray}
where $K_i's$ are some operators, called Kraus operators of $\Phi$. Trace preservation of $\Phi$ is equivalent to the sum $\sum_{i=1}^N K_i^{\dagger}K_i$ equaling the identity operator. The number of Kraus operators $N$ is no more than $\dim\mathcal{H}_1\dim\mathcal{H}_2$, and the minimum number of $N$ is called the Kraus rank of $\Phi$. In some sense, the smaller the Kraus rank is, the simpler the channel is.

The concept of channel consistency is quite intuitive. Consider the following scenario, in which there is a channel from some large system $\mathcal{H}_{1}$ to some large system $\mathcal{H}_{2}$. Here, the mapping is from  $B(\mathcal{H}_{1})$ to $B(\mathcal{H}_{2})$.  A local observer Alice can only gather information from part of $\mathcal{H}_{1}$, say $\mathcal{H}_{1}^{A}$ and part of $\mathcal{H}_{2}$, say $\mathcal{H}_{2}^{A}$; therefore, she has information about the partial mapping from $B(\mathcal{H}_{1}^{A})$ to $B(\mathcal{H}_{2}^{A})$.  Another observer Bob has his information about the partial mapping from some $B(\mathcal{H}_{1}^{B})$ to some  $B(\mathcal{H}_{2}^{B})$. So do any other observers. Several questions naturally arise. How much information about the global quantum channel can be known when Alice, Bob and other observers disclose their local information? If every observer has a description of some partial mapping, is there a global channel satisfying all these local constraints? Can we find some simple channel satisfying local constraints if the local descriptions are known to be compatible?

There are some subtle differences between state consistency and channel consistency. The most confusing part is, how to describe part of a quantum channel, or the partial mapping from a local system to another local system? There is no doubt that part of a quantum state $\rho$, described by  applying some partial trace on $\rho$,  is definitely positive-semidefinite and trace $1$. Therefore, part of a quantum state is again a quantum state. However, in the quantum channel setting, the analogue of the above property is not so straightforward. One may even doubt part of a quantum channel may not be a quantum channel at all. Thanks to the channel-state duality, we can define sub-channel of some quantum channel as the following. 

Given any channel $\Psi:B(\mathcal{H}_{1})\rightarrow B(\mathcal{H}_{2})$, we can always write a corresponding state of $\Psi$ to be
\begin{eqnarray}
\sigma_{\Psi}=\frac{1}{\dim \mathcal{H}_1}\sum_{p,q=1}^{\dim\mathcal{H}_{1}}\ket{p}\bra{q}\otimes \Psi(\ket{p}\bra{q})
\end{eqnarray}
where $\{\ket{i}\}_{i=1}^{\dim\mathcal{H}_{1}}$ is an orthonormal basis of $\mathcal{H}_{1}$.
The state $\sigma_{\Psi}$ is called the Choi-Jamiolkowski state of $\Psi$ and the association above defines an isomorphism between linear maps from $B(\mathcal{H}_{1})$ to $B(\mathcal{H}_{2})$ and operators in $B(\mathcal{H}_{1}\otimes \mathcal{H}_{2})$, called the Choi-Jamiolkowski isomorphism. Its rank is equal to the Kraus rank of $\sigma_{\Psi}$.

Therefore, for any quantum channel  $\Psi: B(\mathcal{H}_{A}\otimes \mathcal{H}_B) \rightarrow B(\mathcal{H}_{A^{\prime}}\otimes \mathcal{H}_{B^{\prime}})$, we can define  channel $\Psi^A: B(\mathcal{H}_A) \rightarrow B(\mathcal{H}_A^{\prime})$ by taking the reduced density operator of the Choi-Jamiolkowski state $\sigma_{\Psi}$ as its Cho-Jamiolkowski state $\sigma_{\Psi^A}$.

Observe that 
\begin{eqnarray}
\Psi^A(\rho)&=&\tr_A \big( (\frac{1}{\dim \mathcal{H}_{AB}}\sum\limits_{p,q=1}^{\dim\mathcal{H}_{AB}}\tr_B\ket{p}\bra{q}\otimes \tr_{B^{\prime}}\Psi(\ket{p}\bra{q}))(\rho\otimes I_{A^{\prime}}) \big )\\
&=&\frac{1}{\dim \mathcal{H}_{AB}}\sum\limits_{p,q=1}^{\dim\mathcal{H}_{AB}} \tr_{AB} \big ( (\ket{p}\bra{q}(\rho\otimes I_B))\otimes \tr_{B^{\prime}}\Psi(\ket{p}\bra{q})\big )\\
&=& \tr_B \Psi(\rho\otimes \frac{I_B}{\dim\mathcal{H}_B}),
\end{eqnarray}
$\Psi^A$ acts exactly the same as $\Psi$ does between $ B(\mathcal{H}_A)$ and $B(\mathcal{H}_A^{\prime})$.  Hence we call $\Psi^A$ sub-channel of $\Psi$ from  $ B(\mathcal{H}_A)$ to $B(\mathcal{H}_A^{\prime})$.


By adopting above definition, we will address the following question: how simple the global channel can be, or more specifically, how small its Kraus rank can be if the sub channels are known to be compatible.
Mathematically,  the local consistency problem for quantum channels can be stated as follows.

\begin{problem}\label{prb:channel}
Consider two multipartite quantum systems $\mathcal{A}_1\mathcal{A}_2\cdots \mathcal{A}_n$ and $\mathcal{B}_{1}\mathcal{B}_{2}\cdots \mathcal{B}_{m}$ with Hilbert spaces $\mathcal{H}_{\mathcal{A}_{1}\mathcal{A}_{2}\cdots \mathcal{A}_{n}}=\mathcal{H}_{\mathcal{A}_{1}}\otimes \mathcal{H}_{\mathcal{A}_{2}}\otimes \cdots \otimes \mathcal{H}_{\mathcal{A}_{n}}$ and $\mathcal{H}_{\mathcal{B}_{1}\mathcal{B}_{2}\cdots \mathcal{B}_{m}}=\mathcal{H}_{\mathcal{B}_{1}}\otimes \mathcal{H}_{\mathcal{B}_{2}}\otimes \cdots \otimes \mathcal{H}_{\mathcal{B}_{m}}$ respectively. Assume a set of local quantum channels $\{\Phi_{l}: l=1,2,\cdots,k\}$ is given. Each $\Phi_l$ maps states on Hilbert space $\mathcal{H}_{I_{l}}=\mathop{\otimes }_{\mathcal{A}_{i}\in I_{l}}\mathcal{H}_{\mathcal{A}_{i}}$ to states on Hilbert space $\mathcal{H}_{J_{l}}=\mathop{\otimes}_{\mathcal{B}_{j}\in J_{l}}\mathcal{H}_{\mathcal{B}_{j}}$ where $I_{l}\subseteq \{\mathcal{A}_{1},\mathcal{A}_{2},\cdots,\mathcal{A}_{n}\}$ and $J_{l}\subseteq \{\mathcal{B}_{1},\mathcal{B}_{2},\cdots,\mathcal{B}_{m}\}$. The local consistency problem for quantum channels is to address the existence of a quantum channel $\Phi:B(\mathcal{H}_{\mathcal{A}_{1}\mathcal{A}_{2}\cdots \mathcal{A}_{n}})\rightarrow B(\mathcal{H}_{\mathcal{B}_{1}\mathcal{B}_{2}\cdots \mathcal{B}_{m}})$  satisfying $\Phi_{l}(\rho)=\tr_{{J}_{l}^{c}}\Phi(\rho\otimes \frac{I}{\dim \mathcal{H}_{I_{l}^{c}} })$ for any $\rho\in D(\mathcal{H}_{I_{l}})$ for any $1\leq l\leq k$.

\end{problem}

By taking the Choi-Jamiolkowski states of each $\Phi_l$, Problem~\ref{prb:channel} can be converted to the existence of the global density operator. Then we can apply Theorem~\ref{thm:main} to the quantum system $\mathcal{A}_1\mathcal{A}_2\cdots \mathcal{A}_n\mathcal{B}_{1}\mathcal{B}_{2}\cdots \mathcal{B}_{m}$ with Hilbert space $\mathcal{H}_{\mathcal{A}_{1}\mathcal{A}_{2}\cdots \mathcal{A}_{n}\mathcal{B}_{1}\mathcal{B}_{2}\cdots \mathcal{B}_{m}}=\mathcal{H}_{\mathcal{A}_{1}}\otimes \mathcal{H}_{\mathcal{A}_{2}}\otimes \cdots \otimes \mathcal{H}_{\mathcal{A}_{n}}\otimes \mathcal{H}_{\mathcal{B}_{1}}\otimes \mathcal{H}_{\mathcal{B}_{2}}\otimes\cdots\otimes \mathcal{H}_{\mathcal{B}_{m}}$.  Since a set of reduced density operators $\sigma_{\Phi_{l}}(l=1,2,\cdots,k)$ is given and each $\sigma_{\Phi_{l}}$ acts on a subsystem $I_{l}J_{l}\subseteq \{ \mathcal{A}_1,\mathcal{A}_2,\cdots, \mathcal{A}_n,\mathcal{B}_{1},\mathcal{B}_{2},\cdots, \mathcal{B}_{m}\}$. According to Theorem~\ref{thm:main}, we can always find a density operator $\sigma$ satisfying the local consistencies and $\textstyle \rank{\sigma}\leq {\sqrt{\sum_l {(\dim{I_l}\dim{J_{l}})}^2}}$.

\begin{theorem}\label{channel}
Assume there exists a global quantum channel $\Phi:B(\mathcal{H}_{\mathcal{A}_{1}\mathcal{A}_{2}\cdots \mathcal{A}_{n}})\rightarrow B(\mathcal{H}_{\mathcal{B}_{1}\mathcal{B}_{2}\cdots \mathcal{B}_{m}})$  satisfying $\Phi_{l}(\rho)=\tr_{{J}_{l}^{c}}\Phi(\rho\otimes \frac{I}{\dim \mathcal{H}_{I_{l}^{c}}})$ for any $\rho\in D(\mathcal{H}_{I_{l}})$ for any $1\leq l\leq k$. Then there also exists a quantum channel $\Phi^{\prime}$ satisfying the same local constraints such that $\Phi^{\prime}$ can be expressed with no more than $\textstyle{\sqrt{\sum_l {(\dim{I_l}\dim{J_{l}})}^2}} $ Kraus operators.
\end{theorem}

\section{Some Examples}\label{sec:example}

\begin{example}
Consider an $n$-qubit quantum system $\mathcal{A}_1\mathcal{A}_2\cdots \mathcal{A}_n$ with Hilbert space $\mathcal{H}_{\mathcal{A}_{1}\mathcal{A}_{2}\cdots \mathcal{A}_{n}}=\mathbb{C}_{2}^{\otimes n}$.  We are interested in the $n$-qubit states $\rho$ such that any $k$-qubit local density operator of $\rho$ is $\frac{I_{k}}{2^{k}}$.

Obviously, $\rho=\frac{I_{n}}{2^{n}}$ is a trivial candidate with the maximal rank $2^{n}$.

From Theorem~$\ref{thm:main}$, there exists some $n$-qubit density operator $\rho$ satisfying same local consistency and $\rank(\rho) \in O(2^{k}\sqrt{{n\choose k}})$. As a corollary, when $k=2$ is fixed, then $\rank(\rho)\in O(n)$.

Indeed, for $k=2$, there are always pure state (i.e. $\text{rank}=1$) solutions for any $n\geq 5$. One such example could be a graph state~\cite{BR00,SW02} $\ket{\psi_n}$ on a ring, which is a common eigenstate of eigenvalue $1$ of the Pauli operators $g_i=\{Z_{i-1}X_iZ_{i+1}\}$ for $i=1,2,\ldots,n$, where $Z_0=Z_n,\ Z_{n+1}=Z_1$. That is,
$g_i\ket{\psi_n}=\ket{\psi_n}$ for $i=1,2,\ldots,n$. Note that
\begin{equation}
\rho_n=\ket{\psi_n}\bra{\psi_n}=\frac{1}{2^n}\prod_{i=1}^n(I+g_i).
\end{equation}
It is then straightforward to see that any $2$-local density operator of the $n$-particle state $\rho_n$ is $\frac{I_{2}}{2^{2}}$.

In general, for any fixed $k$, there do exist $n$-qubit graph states such that any $k$-local density operator of the graph state is $\frac{I_{k}}{2^{k}}$, for large enough $n$~\cite{Dan05}.

\end{example}

\begin{example}
Consider a system of $N$ bosons where each particle has $2$ energy levels.
The $2$-boson maximally mixed state is defined as
\begin{eqnarray}
M_{B}^{(2)}=\frac{1}{3}(\ket{00}\bra{00}+\ket{11}\bra{11}+\ket{\frac{01+10}{\sqrt{2}}}\bra{\frac{01+10}{\sqrt{2}}}).
\end{eqnarray}
Obviously there exists a non-degenerate $N$-boson maximally mixed operator $M_B^{(N)}$ such that $\tr_{3,\cdots,N}(M_B^{(N)})=M_{B}^{(2)}$.
Then it follows from Theorem~\ref{thm:boson} that there exists another $N$-boson density operator
$\sigma$  with  $\tr_{3,\cdots,N}(\sigma)=M_B^{(2)}$ and $\rank{\sigma}\leq 3$.

We can choose $\frac{N-1}{3}\leq p\leq \frac{2N+1}{3}$ and let  $\sigma_{p}$ to be
\begin{eqnarray}
\frac{3p+1-N}{6p}\ket{0\cdots 0}\bra{0\cdots 0}+\frac{2N-3p+1}{6(N-p)}\ket{1\cdots 1}\bra{1\cdots 1}+\frac{(\sum_{i_{1}+\cdots+ i_{N}=p}\ket{i_{1}i_{2}\cdots i_{N}})(\sum_{i_{1}+\cdots+ i_{N}=p}\bra{i_{1}i_{2}\cdots i_{N}})}{6{N-2\choose p-1}}.
\end{eqnarray}

Notice that
\begin{eqnarray}
&&\sum_{i_{1}+i_{2}+\cdots+i_{N}=p}\ket{i_{1}i_{2}\cdots i_{N}}\\
&=&\ket{00}\sum_{i_{3}+\cdots+i_{N}=p}\ket{i_{3}\cdots i_{N}}+(\ket{01}+\ket{10})\sum_{i_{3}+\cdots+i_{N}=p-1}\ket{i_{3}\cdots i_{N}}+\ket{11}\sum_{i_{3}+\cdots+i_{N}=p-2}\ket{i_{3}\cdots i_{N}},
\end{eqnarray}

and then
\begin{eqnarray}
&&\tr_{3,\cdots,N}(\sigma_{p})\\
&=&(\frac{3p+1-N}{6p}+\frac{ {N-2\choose p} }{ 6 {N-2\choose p-1} })\ket{00}\bra{00}+\frac{{N-2 \choose p-1}}{6{N-2 \choose p-1}} \ket{01+10}\bra{01+10}+(\frac{2N-3p+1}{6(N-p)}+\frac{{N-2\choose p-2}}{{N-2\choose p-1}})\ket{11}\bra{11}\\
&=&\frac{1}{3}(\ket{00}\bra{00}+\ket{11}\bra{11}+\ket{\frac{01+10}{\sqrt{2}}}\bra{\frac{01+10}{\sqrt{2}}}).
\end{eqnarray}

Therefore, $\{\sigma_{p}\}_{\frac{N-1}{3}\leq p\leq \frac{2N+1}{3}}$ is a family of $N$-boson density operators with rank $3$ and every $2$-local density operator of any $\sigma_{p}$ is the bosonic maximally mixed state $M_{B}^{(2)}$. Furthermore,  when $N\equiv 1(\mathop{\rm mod 3})$, we will have $\rank(\sigma_{\frac{N-1}{3}})=2$.

\end{example}

\section{Conclusion and Future Works}\label{sec:conclusion}

In this paper we addressed the problem of how simple a solution can be for any given local consistency instance. More specifically, how small the rank of a global density operator can be if the local constraints are known to be compatible. We provided a reduction based approach to this problem and proved that any compatible local density operators can be satisfied with a global density operator with bounded rank. Then we studied both fermionic and bosonic versions of the $N$-representability problem as applications. After applying the channel-state duality, we proved that  any compatible local channels can be satisfied with a global quantum channel which can be expressed with a small number of Kraus operators.

This paper represents a preliminary step toward understanding the structure of solutions to the local consistency problem. There are many open questions from this approach deserving further investigation. For example, though the local consistency is known to be QMA-complete in general,  efficient algorithms are still possible for some classes of instances. Since we know now the existence of solutions is equivalent to the existence of simple solutions, we can ask if it is  possible to find more efficient algorithms for these classes? Further, we could ask if only spectra or other descriptions are known for subsystems, how simple can a solution  be?


\subsection*{Acknowledgements}
JC is supported by NSERC and NSF of China(Grant No. 61179030). ZJ acknowledges support from NSERC, ARO and NSF of China (Grant Nos. 60736011 and
60721061). DWK is supported by NSERC Discovery Grant 400160, NSERC Discovery Accelerator Supplement 400233 and Ontario Early Researcher Award 048142. BZ is supported by NSERC Discovery Grant 400500 and CIFAR. Part of this work was done when JC was a PhD student with Prof. Mingsheng Ying in Tsinghua University. We thank Mingsheng Ying, Aram Harrow, John Watrous, Mary Beth Ruskai for very delightful discussions.


\end{document}